\documentclass[conference,a4paper]{IEEEtran}

\usepackage{epsf}
\usepackage{algorithm,makecell,algorithmic,amsbsy,amsmath,amssymb,epsfig,bbm,mathrsfs,multirow,amsthm}
\usepackage{array,multirow,graphicx}
\usepackage[english]{babel}
\usepackage[justification=centering]{caption}
\usepackage{url}
\usepackage{threeparttable}
\usepackage{subfigure,bm}
\usepackage{placeins}
\usepackage{url}
\usepackage{cite}
\usepackage{comment}

\newtheorem{theorem}{\textbf{\textsc{Theorem}}}

\begin{document}
\bibliographystyle{IEEE2}

\title{Competitive Data Trading in Wireless-Powered Internet of Things (IoT) Crowdsensing Systems with Blockchain\vspace{-3mm}}

\author{
Shaohan Feng$^1$, Wenbo Wang$^1$, Dusit Niyato$^1$, Dong In Kim$^2$, and Ping Wang$^3$ \\
$^1$ School of Computer Engineering, Nanyang Technological University (NTU), Singapore	\\
$^2$ School of Information and Communication Engineering, Sungkyunkwan University (SKKU), Korea		\\
$^3$ Department of Electrical Engineering and Computer Science, York University, Canada \vspace{-5mm}	}

\maketitle

\begin{abstract}
With the explosive growth of smart IoT devices at the edge of the Internet, embedding sensors on mobile devices for massive data collection and collective environment sensing has been envisioned as a cost-effective solution for IoT applications. However, existing IoT platforms and framework rely on dedicated middleware for (semi-) centralized task dispatching, data storage and incentive provision. Consequently, they are usually expensive to deploy, have limited adaptability to diverse requirements, and face a series of data security and privacy issues. In this paper, we employ permissionless blockchains to construct a purely decentralized platform for data storage and trading in a wireless-powered IoT crowdsensing system. In the system, IoT sensors use power wirelessly transferred from RF-energy beacons for data sensing and transmission to an access point. The data is then forwarded to the blockchain for distributed ledger services, i.e., data/transaction verification, recording, and maintenance. Due to coupled interference of wireless transmission and transaction fee incurred from blockchain's distributed ledger services, rational sensors have to decide on their transmission rates to maximize individual utility. Thus, we formulate a noncooperative game model to analyze this competitive situation among the sensors. We provide the analytical condition for the existence of the Nash equilibrium as well as a series of insightful numerical results about the equilibrium strategies in the game.
\end{abstract}


\section{Introduction}
\label{sec:introduction}

At the dawn of 5G, the world has seen an enormous increase in the number of pervasively connected IoT devices, which are used in a plethora of scenarios such as vehicular networks, the logistics/manufacturing sectors, smart homes and e-health. With the trend of sensor miniaturization and the widespread adoption of IPv6, Cisco predicts that by 2021 an extraordinary amount of 847ZB data will be generated by IoT devices annually and about 7.2ZB will be finally stored worldwide~\cite{networking2016cisco}. Such technological development has created unprecedented opportunities of access to ubiquitous sensing data about the context in concern, e.g., smart city and urban environment, for both real-time use and big data-based analysis. However, it also imposes great challenges to network operation and data processing. Compared with the conventional Wireless Sensor Networks (WSNs), most of the IoT sensors are owned by users instead of operators. Meanwhile, the data generated by the same sensors may be consumed by different data services, which require various levels of data quality, timeliness and sampling frequency for different purposes. For this reason, conventional WSNs is limited in proliferation due to the cost of deployment and maintenance as well as the rigidness with task-specified data processing/dispatching structures.

To overcome the limitations of conventional WSNs in both network operation and data processing, a number of novel paradigms have been proposed at both the network side and the data processing side. In~\cite{8353364}, a framework of wireless-powered sensing systems was proposed to address the issues of limited temporal-spatial coverage in urban crowdsensing over wearables. Therein, the mobile operator deploys ultra-dense charging stations using energy beamforming in small cells, which only require incremental upgrade of the protocols running on existing infrastructure. Such design enables the accommodation of the massive-scale, already-in-field IoT devices for Radio Frequency (RF)-powered pervasive sensing. Power transfer is used as incentiveness for IoT devices to execute tasks of crowdsensing applications. For involved parties, this framework helps to form the basis of an energy-data market at the data collection stage. At the operator side, data processing/aggregation is usually delegated to the cloud-based backend~\cite{6069707} or semi-centralized edge servers~\cite{8272334}. Under this paradigm, interconnections between the sensor cloud and the data processing backend rely on an intermediate layer provided by the operator for handling the tasks of data mediation, such as task association, data filtering, privacy preserving and data-integrity verification~\cite{6069707}. However, with the presence of such a middleware pre-designed and fully controlled by the operator, the IoT platforms and crowdsensing framework face the same problem of lacking adaptability and high installation cost as in WSNs. Also, the centralization of data processing, storage and trading inevitably causes the security risks such as data falsification and manipulation because of a single breach.

To overcome the flaws and vulnerability caused by centralization at the data processing stage, we resort to the emerging technology of permissionless blockchains~\cite{wang2018survey} to ensure that the task assignment, the data collection, storage and trading are all performed in a decentralized but trusted manner. Then, the intermediate layer can be safely removed to enhance privacy and data security while the data integrity is still publicly verifiable. In brief, a permissionless blockchain system can be seen as a replicated database maintained by a number of pseudonymous nodes over peer-to-peer (P2P) connections. Blockchains use the public key infrastructure (PKI) mechanism and the data structure of hash linked list to ensure that the time order and the content of a data record (also known as a transaction) cannot be tampered without being noticed once it is confirmed on the chain~\cite{wang2018survey}. The blockchain relies on Byzantine fault-tolerate mechanisms, e.g., Nakamoto protocol~\cite{nakamoto2008bitcoin} or Practical Byzantine Fault Tolerance (PBFT) protocol~\cite{Castro:2002:PBF:571637.571640}, to coordinate the Byzantine agreement, i.e., peer consensus, about the state of the transaction storage among the consensus nodes. For permissionless blockchains, the messaging complexity for reaching the consensus among the P2P nodes are expected to be sufficiently low such that the size of the consensus network scales well.

In this paper, we propose a novel framework of an RF-powered IoT crowdsensing system. The IoT system under consideration involves three parties, i.e., the massive-scale IoT devices/sensors working as a sensing cloud, the wireless network operator, and the permissionless blockchain network as shown in Fig.~\ref{fig:system_model}. The sensors operate by using wireless power transferred from RF-energy beacons. The power is used for data sensing by the sensors and data transmission to the access point. The network operator provides wireless power transfer facility, i.e., RF-energy beacons, and data communication services, i.e., an access point, for the sensors. As such, the sensors are charged with a certain price by the network operator. Afterwards, the sensors send their data to the blockchain for distributed ledger functions, e.g., data verification, logging, and executing a set of actions using smart contracts in the form of programmable automata on the chain~\cite{clack2016smart}. The blockchain network adapts the scheme of Proof-of-Work (PoW)~\cite{nakamoto2008bitcoin} for Sybil attack prevention~\cite{wang2018survey}. The blockchain also charges a transaction fee for processing the data from sensors. Finally, the sensors can trade their verified and secured data in the blockchain to data consumers, gaining a certain revenue. Note that here the data trading processes are completely self-organized with no middle layer controlled by the operator. From the perspective of the sensors and the data consumers, the blockchain can be regarded as a decentralized platform as a service.

By choosing the data transmission rates, the sensors are rational and self-interested in maximizing their utility which is defined as a function of the revenue from trading their data, the price paid to the network operator and the transaction fee paid to the blockchain network. However, this leads to a competitive situation because of coupled interference in wireless communication and transaction fee determination, which are functions of the transmission rate of each sensor. To study this situation, we therefore introduce a noncooperative game model. The following key properties are guaranteed in our proposed IoT system:
\begin{itemize}
  \item The data, i.e., transaction, throughput of the blockchain scales well such that the massive data volume from the sensor cloud is handled smoothly.
  \item The rational and self-interested sensors noncooperatively decide on their own sensing/transferring data for individual utility optimization.
  \item The sensors are able to reach a system equilibrium in a self-organized manner with limited coordination among themselves.
\end{itemize}

\begin{figure}[!t]
\centering
\includegraphics[width=.45\textwidth,trim=40 315 150 340,clip]{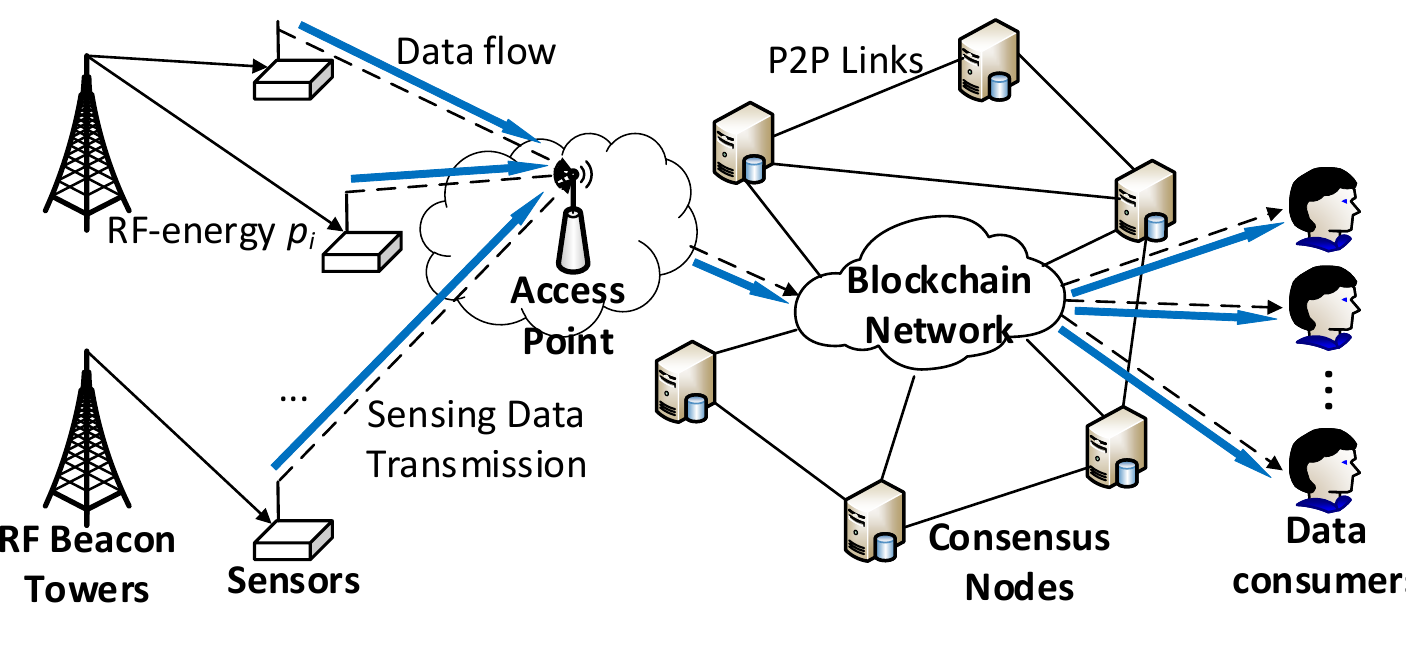}
\caption{Schematics of the RF-powered IoT crowdsensing system.}
\label{fig:system_model}
\end{figure}

\section{System Model}
\label{sec:system}

We consider an RF-powered IoT crowdsensing system as shown in Fig.~\ref{fig:system_model}. Specifically, a cloud of sensors, the set of which is denoted by ${\cal{N}}$, harvest energy from the beamforming-enabled RF-energy beacons deployed by the network operator to support data transfer from the sensors to the access point. Each sensor $i\! \in\! {\cal{N}}$ independently negotiates with the operator about the wirelessly received power level $p_i$ within the range ${\cal{D}}_{p_i}\!=\!\left[0, p^{\rm{u}}_i\right]$. Let $c_i$ be the fixed power level used by sensor $i$ for circuit maintenance and data sensing. Then, the power used by sensor $i$ to transmit data for storage and trading is $p_i - c_i$. We assume that the propagation model for data transmission of sensor $i$ follows~(\ref{eq:transmission_rate}) and consider the path loss as a function of the distance $d_i$ between sensor $i$ and the access point connected to the blockchain. Let ${\mathbf{r}}=\left[r_{i}\right]^{\top}_{i \in {\cal{N}}}$ be the vector of the transmission rates, ${\bf{p}}=\left[p_i\right]^{\top}_{i \in {\cal{N}}}$ be the vector of received powers, i.e., the power transferred from the RF-energy beacons received by a sensor. Likewise, ${\bf{p}}_{-i}$ be the vector of the received powers except sensor $i$, the transmission rate of sensor $i$ can be derived as follows $\forall i \in {\cal{N}}$:
\begin{eqnarray}\label{eq:transmission_rate}
r_i={R_i}\left( {{p_i},{{\bf{p}}_{ - i}}} \right) = {b_i}{\log _2}\left( {1 + \frac{g_i{\frac{{{p_i} - {c_i}}}{{{{\left( {{d_i}} \right)}^{{\alpha _i}}}}}}}{{\sum\limits_{j \ne i} g_j {\frac{{{p_j} - {c_j}}}{{{{\left( {{d_j}} \right)}^{{\alpha _j}}}}}}  + {\sigma ^2}}}} \right),\,
\end{eqnarray}
where $\sigma^2$ is the variance of the additive white Gaussian noise, $g_i$ is the channel gain of sensor $i$, and $b_i$ is the corresponding bandwidth. We define ${\bf{R}}\left({\bf{p}}\right)\!=\!\left[{R_i}\left( {{p_i},{{\bf{p}}_{ - i}}} \right)\right]^{\top}_{i \in {\cal{N}}}$, where ${\bf{R}}\left({\bf{p}}\right)\!:\! {\cal{D}}_{\bf{p}}\! =\!\mathop \times_{i\in {\cal{N}}} {\cal{D}}_{p_i} \mapsto {\cal{D}}_{\bf{r}}$ is a continuous closed mapping.

We assume that the data services, e.g., data storage, trading and task dispatching, are implemented on top of a permissionless blockchain. The sensing data are formated into normal transactions of fixed size. To enhance efficiency, only the digest of each transaction is stored on the chain, and the content of the transactions are stored by each consensus node off-chain. To address the issue of transaction throughput bottleneck in Bitcoin-like blockchains, we adopt the protocol of ELASTICO for sharded blockchain in~\cite{luu2016secure}, where the transaction throughput increases linearly with the computational power admitted by the blockchain. Compared with the existing protocols based on PoW (see~\cite{wang2018survey} and the references therein), the sharded blockchain employs the process of PoW puzzle solving for identity establishment and forms a number of consensus committees to handle the received transactions in disjoint subsets through standard Byzantine agreement protocols. Therefore, it is possible to scale the transaction throughput of the blockchain according to the total transmission rate of the sensors as long as the sensors are able to sustain the maintenance cost of the blockchain.

We consider that the maintenance cost of the blockchain is measured in the supplied computational power. Based on the linear relationship between the computational power and the transaction throughput, we can further express the required computational power as ${m\sum_{j \in {\cal{N}}} {{R_j}\left( {{p_j},{{\bf{p}}_{ - j}}} \right)} }$, where $m\!>\!0$ is the computational power coefficient and $\sum_{j \in {\cal{N}}} {{R_j}\left( {{p_j},{{\bf{p}}_{ - j}}} \right)}$ is the total data rate of the sensor cloud. Furthermore, it is well-known that the power consumption in modern computer architecture can be modeled as a quadratic function of the corresponding computational frequency~\cite{teodorescu2008variation}. Therefore, the power consumption of the sharded blockchain network can be derived as
\begin{equation}\label{eq:blockchain_computing_power}
p_{\mathrm{b}} \!=\!{a{{\left[{m\sum_{j \in {\cal{N}}} {{R_j}\left( {{p_j},{{\bf{p}}_{ - j}}} \right)} } \right]}^2} \!+\! b\left[ {m\sum_{j \in {\cal{N}}} {{R_j}\left( {{p_j},{{\bf{p}}_{ - j}}} \right)} } \right] \!+\! c},
\end{equation}
where $a$, $b$ and $c$ are the power consumption coefficients.

To enjoy the data service (e.g., smart contracts and transaction recording) provided by the sharded blockchain, sensor $i$ needs to pay the transaction fee to compensate the cost of the consensus nodes incurred by energy consumption. We assume that the sensors proportionally share the blockchain maintenance cost among themselves according to the volume of data that they propose to the blockchain. Then, the rate of payment by sensor~$i$ depends on its fraction of the total transmit rate as follows:
\begin{equation}\label{eq:blockchain_cost}
\begin{array}{ll}
C_i=&\frac{{R_i}\left({p_i},{{\bf{p}}_{ - i}} \right)}{\sum_{j\in {\cal{N}}}{R_j}( {p_j},{{\bf{p}}_{ - j}} )}\left\{a{{\left[ m\sum_{j \in {\cal{N}}} {{R_j}\left( {{p_j},{{\bf{p}}_{ - j}}} \right)}  \right]}^2} \right.\\
&\left.+ b\left[ m\sum_{j \in {\cal{N}}} {{R_j}( {{p_j},{{\bf{p}}_{ - j}}} )} \right] + c\right\}.
\end{array}
\end{equation}

The revenue of a sensor is a function of the volume of its sensing data recorded by the blockchain. Let $\lambda_i$ denote the price of unit bitrate of sensor $i$ and $\phi$ denote the price of unit power transferred from the RF-energy beacons. To support the wirelessly received power level $p_i$ for sensor $i$, the wirelessly transferred power level of the RF-energy beacons $i$ is $P^{\rm{t}}\left(p_i, d^{\rm{t}}_i\right)=p_i\left( d^{\rm{t}}_i\right)^\eta$ based on the Slivnyak-Mecke��s theorem~\cite{baccelli2010stochastic}, where $\eta$ is the path-loss exponent of wirelessly power transfer and $d^{\rm{t}}_i$ is the distance between the RF-energy beacons and sensor $i$. Based on (\ref{eq:transmission_rate})-(\ref{eq:blockchain_cost}), the utility of sensor $i$ can be expressed as follows:
\begin{equation}\label{eq:sensor_i}
\small{\begin{array}{ll}
{u_i}\left( {{p_i},{{\bf{p}}_{ - i}}} \right) = \lambda_i {R_i}\left( {{p_i},{{\bf{p}}_{ - i}}} \right) - \phi P^{\rm{t}}\left(p_i, d^{\rm{t}}_i\right) - \frac{{R_i}\left( {{p_i},{{\bf{p}}_{ - i}}} \right)}{\sum_{j\in {\cal{N}}}{R_j}\left( {{p_j},{{\bf{p}}_{ - j}}} \right)}\\
\times\left\{ {a{{\left[ {m\sum\limits_{j \in {\cal{N}}} {{R_j}\left( {{p_j},{{\bf{p}}_{ - j}}} \right)} } \right]}^2} \!+\! b\left[ {m\sum\limits_{j \in {\cal{N}}} {{R_j}\left( {{p_j},{{\bf{p}}_{ - j}}} \right)} } \right] \!+\! c} \right\}.
\end{array}}
\end{equation}
In (\ref{eq:sensor_i}), a larger $\lambda_i$ indicates a higher quality (hence a higher value) of the data generated by sensor $i$.

Each sensor aims to maximize its individual utility, i.e., $\max\limits_{p_i\in {\cal{D}}_{p_i}}{u_i}\left( {{p_i},{{\bf{p}}_{ - i}}} \right)$, given the interference from other sensors. Therefore, the sensors' strategy are coupled, and a noncooperative game can be formulated as a four-tuple ${\cal{G}}_{\rm{s}}=\left\{{\cal{N}}, {\bf{p}}, {\cal{D}}_{\bf{p}}, {\bf{u}}\right\}$, where
\begin{itemize}
\item ${\cal{N}}$ is the set of active sensors, i.e., the players;

\item ${\cal{D}}_{\bf{p}} \subset {\mathbb{R}}^{\left|{\cal{N}}\right|}$ is the domain of ${\bf{p}}$, i.e., strategy, and an $|\mathcal{N}|$-polyhedron;

\item ${\bf{p}}\in{\cal{D}}_{\bf{p}}$ is the sensor-determined received power vector;

\item ${\bf{u}}=\left[u_i\right]_{i \in {\cal{N}}}$ is the vector of sensors' utilities as a function of ${\bf{p}}$,  where $u_i$ is given by~(\ref{eq:sensor_i}).
\end{itemize}

Based on the game formulation, we consider the Nash equilibrium to be the solution for the sensors.


\section{Equilibrium Analysis}
\label{sec:analysis}

To ease the analysis of the Nash Equilibrium (NE) of ${\cal{G}}_{\rm{s}}$, we consider that the sensors optimize their utilities by deciding on their transmission rates instead of on their strategies ${\bf{p}}$. This is owning to the fact that the vector of the functions describing the relationship between ${\bf{r}}$ and ${\bf{p}}$, i.e., ${\bf{r}}={\bf{R}}\left({\bf{p}}\right)$, is a continuous, closed injective operator as shown in Theorem~\ref{th:one_to_one}. Then, there exists an inverse operator of ${\bf{R}}\left({\bf{p}}\right)$, denoted by ${\bf{R}}^{-1}\left({\bf{r}}\right): {\cal{D}}_{\bf{r}} \mapsto {\cal{D}}_{\bf{p}}$ such that the sensors' received powers ${\bf{p}}$ is determined given their transmission rates ${\bf{r}}$.

\begin{theorem}
\label{th:one_to_one}
There exists an inverse operator of ${\bf{R}}\left({\bf{p}}\right)$, i.e., ${\bf{R}}^{-1}\left({\bf{r}}\right): {\cal{D}}_{\bf{r}} \mapsto {\cal{D}}_{\bf{p}}$, such that ${\bf{p}} = {\bf{R}}^{-1}\left({\bf{r}}\right)$.
\end{theorem}

\begin{proof}
Let $\gamma_i\left(r_i\right)={e^{\frac{{{r_i}\ln 2}}{{{b_i}}}}} - 1$, $\beta_i\left(p_i\right)={{{g_i}\frac{{{p_i} - {c_i}}}{{{{\left( {{d_i}} \right)}^{{\alpha _i}}}}}}}$, and ${\bm{\beta}}\left({\bm{p}}\right)=\left[\beta_i\left(p_i\right)\right]^{\top}_{i \in {\cal{N}}}$. According to~(\ref{eq:transmission_rate}), we have
\begin{equation}
\gamma_i\left(r_i\right) = \Lambda_i\left({\bm{\beta}}\left({\bm{p}}\right)\right)= \frac{\beta_i\left(p_i\right)}{{\sum\limits_{j \ne i} \beta_j\left(p_j\right) + {\sigma ^2}}}.
\end{equation}
Since $\forall i\! \in\! {\cal{N}}$ both $\gamma_i\left(r_i\right)$ and $\beta_i\left(p_i\right)$ are continuous, closed injective operators, the injective properties of  ${\bf{R}}\left({\bf{p}}\right)$ can be ensured iff ${\bm{\Lambda}}\left({\bm{\beta}}\left({\bm{p}}\right)\right)=\left[ \Lambda_i\left({\bm{\beta}}\left({\bm{p}}\right)\right) \right]_{i \in {\cal{N}}}$ is injective.

We prove by contradiction that ${\bm{\Lambda}}\left({\bm{\beta}}\left({\bm{p}}\right)\right)$ is injective. Assume that ${\bm{\Lambda}}\left({\bm{\beta}}\left({\bm{p}}\right)\right)$ is not injective. Then, there exist $ {\bf{p}}'$, ${\bf{p}} \!\in\! {\cal{D}}_{\bf{p}}$ and ${\bf{p}}' \!\ne\! {\bf{p}}$ such that ${\bf{r}}' \!=\! {\bf{r}}$. Without loss of generality, we assume $p'_i \!>\! p_i$. To ensure $r'_i=r_i$, ${\bf{p}}'$ should satisfy
\begin{equation}
\frac{\beta_i\left(p'_i\right)}{{\sum\limits_{j \ne i} \beta_j\left(p'_j\right) + {\sigma ^2}}} = \frac{\beta_i\left(p_i\right)}{{\sum\limits_{j \ne i} \beta_j\left(p_j\right) + {\sigma ^2}}},
\end{equation}
and hence
\begin{equation}\label{eq:contradiction}
\sum\limits_{j \ne i} {{\beta _j}\left( {{{p'}_j}} \right)}  = \frac{{{\beta _i}\left( {{{p'}_i}} \right)}}{{{\beta _i}\left( {{p_i}} \right)}}\sum\limits_{j \ne i} {{\beta _j}\left( {{p_j}} \right)}  + \left[ {\frac{{{\beta _i}\left( {{{p'}_i}} \right)}}{{{\beta _i}\left( {{p_i}} \right)}} - 1} \right]{\sigma ^2}.
\end{equation}
By~(\ref{eq:contradiction}), for $\forall l \!\in\! {\cal{N}}$, we have the following equality:
\begin{equation}\label{eq:contradiction_quote}
\begin{array}{ll}
\frac{{{\beta _l}\left( {{{p'}_l}} \right)}}{{\sum\limits_{j \ne l} {{\beta _j}\left( {{{p'}_j}} \right)}  + {\sigma ^2}}} = \frac{{{\beta _l}\left( {{{p'}_l}} \right)}}{{{\beta _i}\left( {{{p'}_i}} \right) + \sum\limits_{j \ne i} {{\beta _j}\left( {{{p'}_j}} \right)}  - {\beta _l}\left( {{{p'}_l}} \right) + {\sigma ^2}}} \\
\mathop=\limits^{(\ref{eq:contradiction})}\frac{{{\beta _l}\left( {{{p'}_l}} \right)}}{{{\beta _i}\left( {{{p'}_i}} \right) + \frac{{{\beta _i}\left( {{{p'}_i}} \right)}}{{{\beta _i}\left( {{p_i}} \right)}}\sum\limits_{j \ne i} {{\beta _j}\left( {{p_j}} \right)}  + \left[ {\frac{{{\beta _i}\left( {{{p'}_i}} \right)}}{{{\beta _i}\left( {{p_i}} \right)}} - 1} \right]{\sigma ^2}  - {\beta _l}\left( {{{p'}_l}} \right) + {\sigma ^2}}} \\
=\frac{{{\beta _l}\left( {{{p'}_l}} \right)}}{{ \frac{{{\beta _i}\left( {{{p'}_i}} \right)}}{{{\beta _i}\left( {{p_i}} \right)}}\left[\sum\limits_{j \in {\cal{N}}} {{\beta _j}\left( {{p_j}} \right)}  +  {\sigma ^2} \right] - {\beta _l}\left( {{{p'}_l}} \right) }}.
\end{array}
\end{equation}

With ${\bf{r}}'\!=\!{\bf{r}}$, we have $r'_l\!=\!r_l$, $\forall l \!\in\! {\cal{N}}$. Then, according to~(\ref{eq:contradiction_quote}), we have the following equality condition:
\begin{equation}\label{eq:intermedaite_equality}
\begin{array}{ll}
\frac{{{\beta _l}\left( {{p_l}} \right)}}{{\sum\limits_{j \ne l} {{\beta _j}\left( {{p_j}} \right)}  + {\sigma ^2}}} \!=\!\frac{{{\beta _l}\left( {{{p'}_l}} \right)}}{{\sum\limits_{j \ne l} {{\beta _j}\left( {{{p'}_j}} \right)}  + {\sigma ^2}}}\!=\!\frac{{{\beta _l}\left( {{{p'}_l}} \right)}}{{ \frac{{{\beta _i}\left( {{{p'}_i}} \right)}}{{{\beta _i}\left( {{p_i}} \right)}}\left[\sum\limits_{j \in {\cal{N}}} {{\beta _j}\left( {{p_j}} \right)}  +  {\sigma ^2} \right] - {\beta _l}\left( {{{p'}_l}} \right) }}  \\
\Leftrightarrow  {\beta _l}\left( {{{p'}_l}} \right) = \frac{{{\beta _i}\left( {{{p'}_i}} \right)}}{{{\beta _i}\left( {{p_i}} \right)}}{\beta _l}\left( {{p_l}} \right).
\end{array}
\end{equation}

\begin{figure*}
\begin{equation}\label{eq:sensor_i_r}
{u_i}\left( {{r_i},{{\bf{r}}_{ - i}}} \right) = \lambda_i {r_i} - \phi P^{\rm{t}}\left(R_i^{ - 1}\left( {\bf{r}} \right), d^{\rm{t}}_i\right)  - \frac{{{r_i}}}{{\sum\limits_{j \in {\cal{N}}} {{r_j}} }}\left[ {a{{\left( {m\sum\limits_{j \in {\cal{N}}} {{r_j}} } \right)}^2} + b\left( {m\sum\limits_{j \in {\cal{N}}} {{r_j}} } \right) + c} \right].
\end{equation}\vspace{-4mm}
\end{figure*}

Summing ${\beta _l}\left( {{{p'}_l}} \right)$ with respect to $l$ over ${\cal{N}}$ leads to $\sum\limits_{l\in {\cal{N}}} {\beta _l}\left( {{{p'}_l}} \right) = \sum\limits_{l\in {\cal{N}}}\frac{{{\beta _i}\left( {{{p'}_i}} \right)}}{{{\beta _i}\left( {{p_i}} \right)}}{\beta _l}\left( {{p_l}} \right)$.
This contradicts with the result derived from~(\ref{eq:contradiction}), i.e., ${\beta _i}\left( {{{p'}_i}} \right) + \sum\limits_{j \ne i} {{\beta _j}\left( {{{p'}_j}} \right)} = \sum\limits_{j\in {\cal{N}}} {{\beta _j}\left( {{{p'}_j}} \right)}  = \frac{{{\beta _i}\left( {{{p'}_i}} \right)}}{{{\beta _i}\left( {{p_i}} \right)}}\sum\limits_{j \in{\cal{N}}} {{\beta _j}\left( {{p_j}} \right)}  + \left[ {\frac{{{\beta _i}\left( {{{p'}_i}} \right)}}{{{\beta _i}\left( {{p_i}} \right)}} - 1} \right]{\sigma ^2}$.
Therefore, the assumption that ${\bm{\Lambda}}\left({\bm{\beta}}\left({\bm{p}}\right)\right)$ is not injective cannot be true, and ${\bm{\Lambda}}\left({\bm{\beta}}\left({\bm{p}}\right)\right)$ is injective. With the injective property of $\gamma_i\left(r_i\right)$ and $\beta_i\left(p_i\right)$, $\forall i \in {\cal{N}}$, ${\bf{R}}\left({\bf{p}}\right)$ is an injective operator. Since ${\bf{R}}\left({\bf{p}}\right)$ is a continuous, closed operator, its inverse operator, i.e., ${\bf{R}}^{-1}\left({\bf{r}}\right)$, exists and the proof is completed.
\end{proof}

By Theorem~\ref{th:one_to_one}, the utility of sensor $i$ in the form of (\ref{eq:sensor_i}) can be rewritten as in (\ref{eq:sensor_i_r}). It is now possible for each sensor to optimize the individual utility by deciding on its transmission rate instead of the received powers from the RF-energy beacons. Then, we can obtain the following condition for the existence of the NE in game ${\cal{G}}_{\rm{s}}$:

\begin{theorem}\label{th:exist}
The NE to the noncooperative game ${\cal{G}}_{\rm{s}}$ exists if $am^2-c\ge0$ and ${\sum\limits_{j \in {\cal{N}}} {{r_j}} } \ge 1$.
\end{theorem}

\begin{proof}
By~(\ref{eq:sensor_i_r}), ${u_i}\left( {{r_i},{{\bf{r}}_{ - i}}} \right)$ is continuous and differentiable on $r_i$, $\forall i \in {\cal{N}}$. Now, we examine the second derivatives of ${u_i}\left( {{r_i},{{\bf{r}}_{ - i}}} \right)$ with respect to $r_i$ as shown in~(\ref{eq:sensor_i_r_second_derivative}), $\forall i \in {\cal{N}}$:
\begin{equation}\label{eq:sensor_i_r_second_derivative}
\begin{array}{ll}
\frac{\partial^2{u_i}\left( {{r_i},{{\bf{r}}_{ - i}}} \right)}{\partial {r_i}^2}\\
\!=\! - \phi \frac{{{\partial ^2 P^{\rm{t}}\left(R_i^{ - 1}\left( {\bf{r}} \right), d^{\rm{t}}_i\right) }}}{{\partial {r_i}^2}} \!-\! 2am^2+{\raise0.7ex\hbox{${2c\sum\limits_{j \ne i} {{r_j}} }$} \!\mathord{\left/
 {\vphantom {{2c\sum\limits_{j \ne i} {{r_j}} } {{{\left( {\sum\limits_{j \in {\cal{N}}} {{r_j}} } \right)}^3}}}}\right.}
\!\lower0.7ex\hbox{${{{\left( {\sum\limits_{j \in {\cal{N}}} {{r_j}} } \right)}^3}}$}}.
\end{array}
\end{equation}
By Theorem~\ref{th:exist}, the sum of the last two terms in~(\ref{eq:sensor_i_r_second_derivative}), i.e., $- 2am^2+{\raise0.7ex\hbox{${2c\sum\limits_{j \ne i} {{r_j}} }$} \!\mathord{\left/
 {\vphantom {{2c\sum\limits_{j \ne i} {{r_j}} } {{{\left( {\sum\limits_{j \in {\cal{N}}} {{r_j}} } \right)}^3}}}}\right.\kern-\nulldelimiterspace}
\!\lower0.7ex\hbox{${{{\left( {\sum\limits_{j \in {\cal{N}}} {{r_j}} } \right)}^3}}$}}$, is smaller than $0$. Moreover, since $r_i = R_i\left(p_i, {\bf{p}}_{-i}\right)$ defined in~(\ref{eq:transmission_rate}) is concave with respect to $p_i$, its inverse operator, i.e., $p_i=R^{-1}_i\left({\bf{r}}\right)$, is correspondingly convex with respect to $r_i$ according to the injective property. Since ${{ P^{\rm{t}}\left(R_i^{ - 1}\left( {\bf{r}} \right), d^{\rm{t}}_i\right) }}$ is a linear function of $R_i^{ - 1}\left( {\bf{r}} \right)$, $- \phi \frac{{{\partial ^2 P^{\rm{t}}\left(R_i^{ - 1}\left( {\bf{r}} \right), d^{\rm{t}}_i\right) }}}{{\partial {r_i}^2}}$ is smaller than $0$, and $\frac{\partial^2}{\partial {r_i}^2}{u_i}\left( {{r_i},{{\bf{r}}_{ - i}}} \right)$ is therefore smaller than $0$. According to Theorem~1 in~\cite{peng2009summary}, the solution to the noncooperative game ${\cal{G}}_{\rm{s}}$ exists, and the proof is completed.
\end{proof}

It is well-known that following the concavity condition in Theorem~\ref{th:exist}, the continuous better-reply in the form of simultaneous gradient ascent admits an equilibrium point (see Theorem 7 in~\cite{peng2009summary}). Due to the space limit, we omit the presentation of the equilibrium searching algorithm and the discussion on the globally asymptotic stability of the equilibrium. Interested readers are referred to~\cite{peng2009summary} for more details.

\section{Performance Evaluation}
\label{sec:performance}

In this section, we present numerical studies to evaluate the performance of the IoT crowdsensing system. For ease of illustration, we consider $10$ sensors, i.e, $\left|{\cal{N}}\right| \!=\! 10$, working as a sensing cloud. The bandwidth of a sensor is $b_i\!=\!2$, $\forall i \!\in\! {\cal{N}}$, and the noise $\sigma$ is $1$. The vector of the distances between the sensors and the access point connected to the blockchain, i.e., ${\bf{d}}\!=\!\left[d_i\right]_{i \in {\cal{N}}}$, is set to be $\left[0.25, 0.2, 0.15, 0.25, 0.2, 0.15, 0.25, 0.2, 0.15, 0.25\right]$, and ${\bm{\alpha}}\!=\!\left[\alpha_i\right]_{i \in {\cal{N}}}\!=\!\left[3.5,3,2.5,3.5,3,2.5,3.5,3,2.5,3.5\right]$. The channel gain ${\bf{g}}=\left[g_i\right]_{i \in {\cal{N}}}$ is set to be $[1.95,2,2.18,1.95,2,2.18,1.95,2,2.18,1.95]$. The price of unit received power $\phi$ is set to be $0.01$, and the circuit and sensing power vector for the sensors
${\bf{c}}= \left[c_i\right]_{i \in {\cal{N}}}$, is $\left[1, 2, 3, 1, 2, 3, 1, 2, 3, 1\right]$. The coefficients in the blockchain power consumption model are $a=0.1$, $b=0.1$, $c=0.1$ and the computational power coefficient is $m=3$. The distance between the sensors and RF-energy beacons ${\bf{d}}^{\rm{t}}=\left[d^{\rm{t}}_i\right]_{i \in {\cal{N}}}$ is $\left[1,2,3,1,2,3, 1, 2, 3, 1\right]\times10^{-3}+1$ and $\eta=2$. The price of unit bitrate of sensors is $\lambda_i=20$, $\forall i \in {\cal{N}}$.

\subsection{Numerical Result}

\begin{figure}[!]
\centering
\begin{minipage}{4.3cm}
\centering
\includegraphics[width=1.04\textwidth,trim=5 5 30 10,clip]{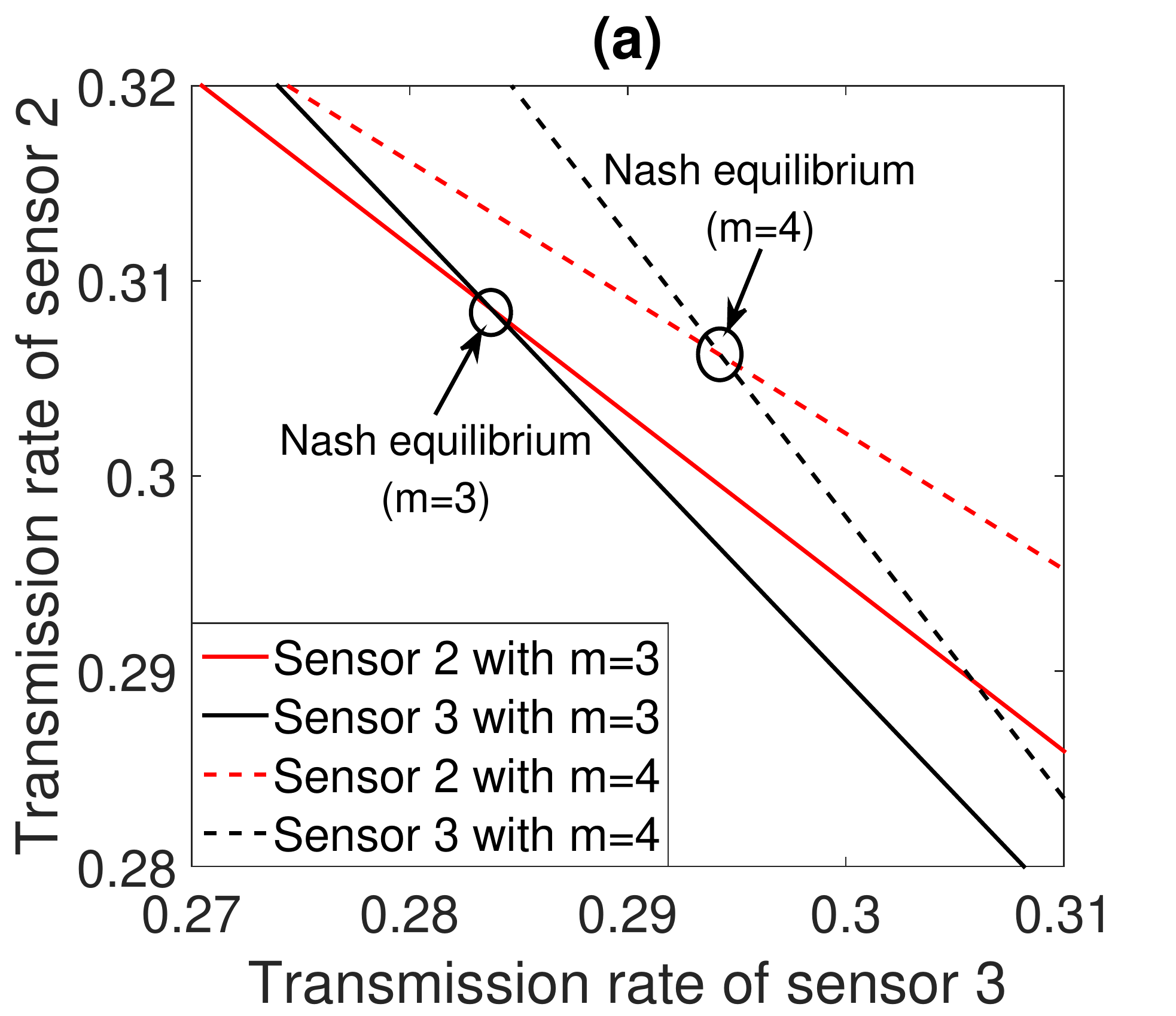}
\end{minipage}
\begin{minipage}{4.3cm}
\centering
\includegraphics[width=1.04\textwidth,trim=5 5 35 0,clip]{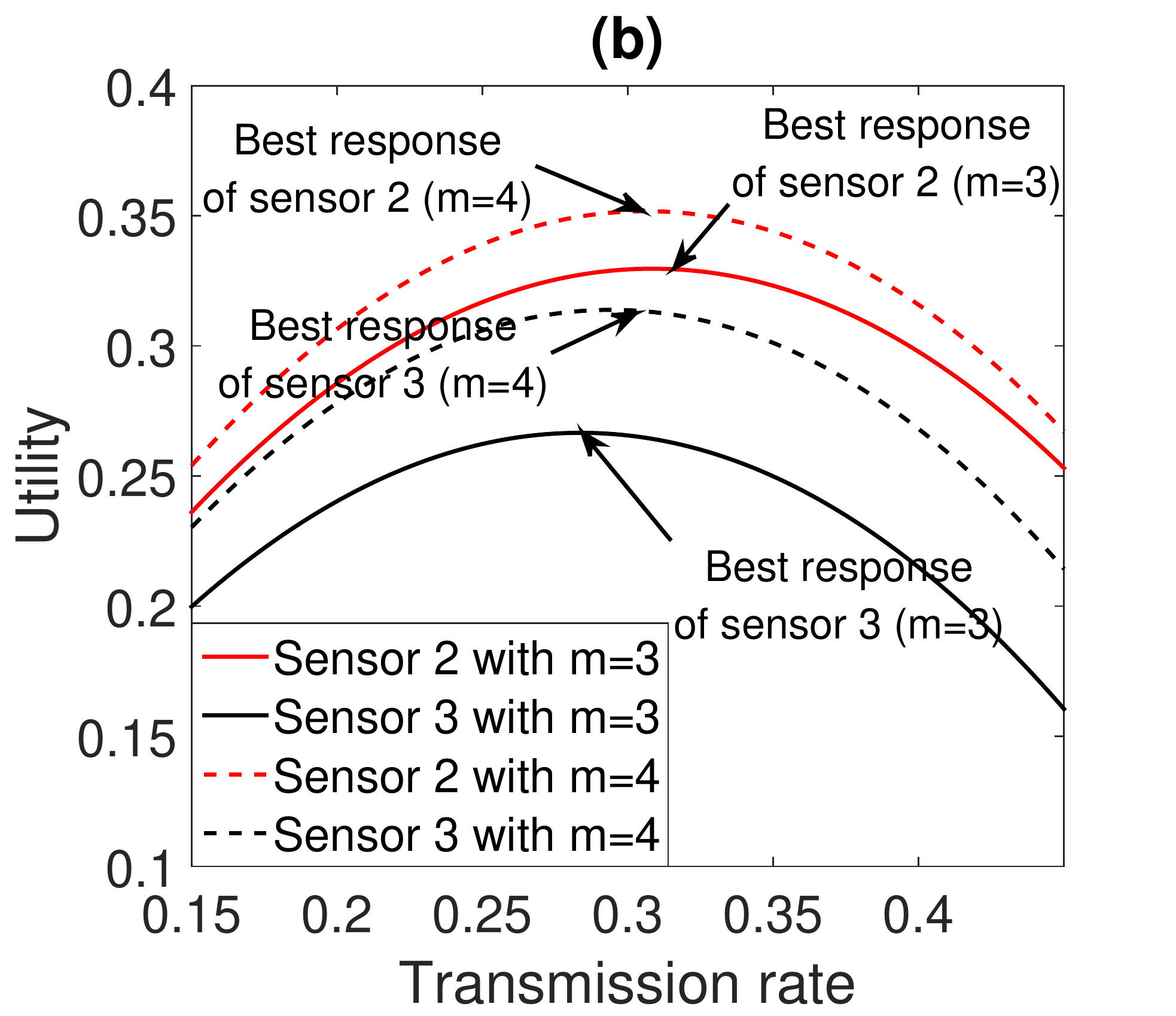}
\end{minipage}
\caption{(a) Nash equilibrium and (b) best response.\vspace{-4mm}}
\label{fig:ne_br}
\end{figure}

Figures~\ref{fig:ne_br}(a) and~(b) show the NE and best responses, respectively. Figure~\ref{fig:ne_br}(a) illustrates the NE of the transmission rates for sensors $2$ and $3$. The NE is the point at which the best responses for sensors $2$ and $3$ intersect. We observe that the transmission rate of sensor $3$ increase as the computational power coefficient $m$ increases. The reason is that when $m$ increases, the computational power increases more quickly, resulting in high energy consumption of the blockchain. Consequently, the transaction fee rises sharply, and hence the less cost-effective transaction fees that sensor $3$ will be charged. In this case, sensor $3$ increase its transmission rate. In contrast, the transmission rate of sensor $2$ decreases as the value of the computational power coefficient $m$ increases. The reason is that the transmission rate of sensor $2$ is higher than that of sensor $3$, which means that the transaction fee of sensor $2$ increases more rapidly than that of sensor $3$. Accordingly, sensor $2$ decreases its transmission rate.

We then evaluate the utilities of sensors $2$ and $3$. In Fig.~\ref{fig:ne_br}(b), the utility of sensor $2$ changes because of the different transmission rates for transferring the data to the blockchain. From Fig.~\ref{fig:ne_br}(b), there is a point where the utility of sensor $2$ is maximized, which is pointed by the arrowhead of ``Best response''. This point indicates the NE for sensor $2$. As is evident from Fig.~\ref{fig:ne_br}(b), this utility, which is a function of the transmission rate, is unimodal, and the optimal solution can be obtained analytically.

\begin{figure}[!]
\centering
\begin{minipage}{4.3cm}
\centering
\includegraphics[width=1.05\textwidth,trim=0 0 50 0,clip]{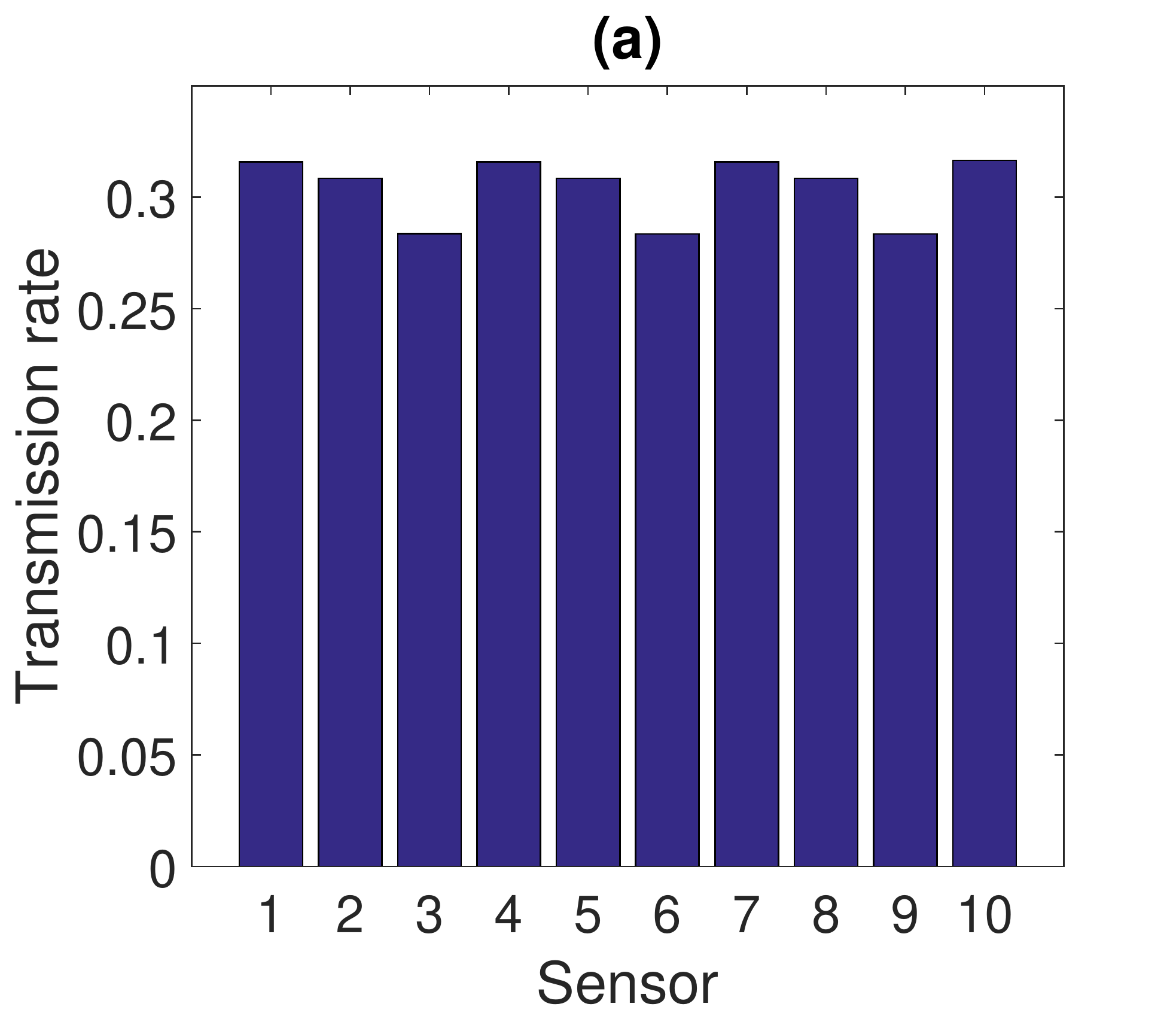}
\end{minipage}
\begin{minipage}{4.3cm}
\centering
\includegraphics[width=0.95\textwidth,trim=100 80 70 0,clip]{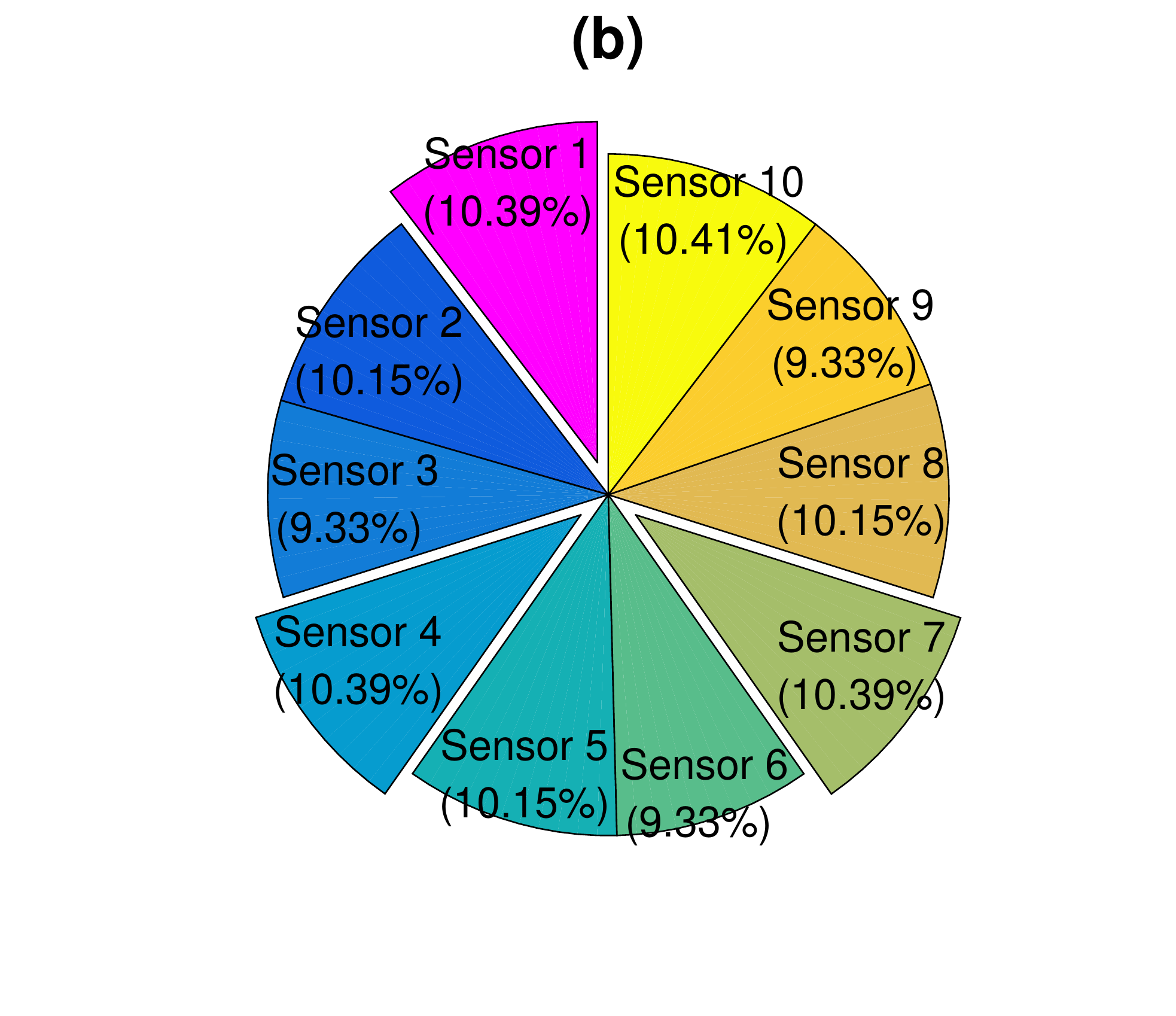}
\end{minipage}
\begin{minipage}{4.3cm}
\centering
\includegraphics[width=1.05\textwidth,trim=0 0 50 0,clip]{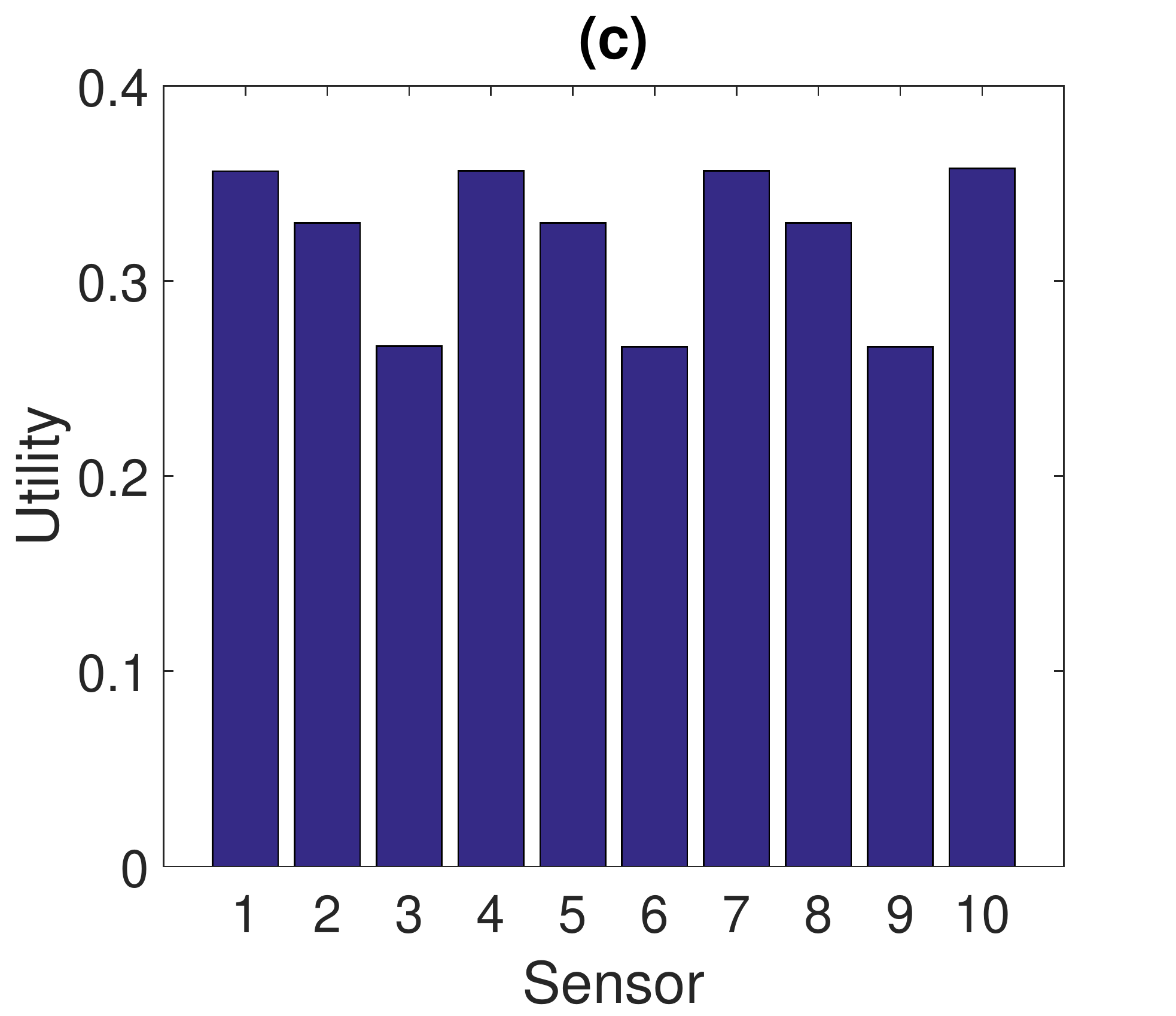}
\end{minipage}
\caption{(a) Transmission rates, (b) ratio of transaction fee, and (c) utilities for each sensor.\vspace{-4mm}}
\label{fig:rate_ratio_utility}
\end{figure}

We next investigate the sensors' states, i.e., the transmission rates, ratio of the total transaction fee, and utilities in Figs.~\ref{fig:rate_ratio_utility}(a), (b), and (c), respectively. As shown in Fig.~\ref{fig:rate_ratio_utility}, the transmission rates of sensors $1$, $4$, $7$, and $10$ are among the highest ones because of their shortest distance to the RF-energy beacons. Specifically for sensors $1$ and $2$, the gap between sensors $1$ and $2$ in the utility is even larger than that in the transmission rate, i.e., $\frac{u_1}{u_2}\approx \frac{0.35}{0.33} > \frac{r_1}{r_2}\approx \frac{0.31}{0.3}$. The reason is that each sensor needs to pay the transaction fee according to its fraction of the sum of all the sensors' transmission rates while the total transaction fee is increasing rapidly, i.e., as a convex function, with respect to the total transmission rates. This means that the less the sensor's fraction of the sum of all the sensors' transmission rates is, the less cost-effective transaction fee that it will be charged. This is consistent with the result shown in Fig.~\ref{fig:ne_br}. 

%

\section{Conclusion}
\label{sec:conclusion}

In this paper, we have presented a noncooperative-game model to analyze the transmission strategy in the self-organized wireless-powered IoT crowdsensing system built upon permissionless blockchains. We have focused on the interactions of the sensors and considered the impact of both the interference and the blockchain maintenance cost on the sensors' utilities. Analytically, we have established a joint model describing the impact of the sensors' transmission strategies on their transmission rate and the needed transaction fee on the blockchain side. We have studied the equilibrium strategies of the sensors in the wireless-powered IoT crowdsensing system by using best response. We have analytically examined the conditions for the solution, i.e., the Nash equilibrium, of the game to exist. Our future work will extend to the study in the impact of user-demands on the sensors' strategies.

\bibliography{bibfile}

\end{document}